\documentclass[12pt,a4paper]{amsart}

\usepackage[latin1]{inputenc}     
\usepackage{amsmath,amssymb,amsthm}
\usepackage{latexsym}
\usepackage{amsfonts}
\usepackage{bbm,bm,dsfont} 
\usepackage{color} 
\usepackage{graphicx}

\numberwithin{equation}{section} 

\theoremstyle{definition}
\newtheorem{proposition}{Proposition}
\newtheorem{definition}{Definition}

\newtheorem{remark}{Remark}
\newtheorem{theorem}{Theorem}
\newtheorem{lemma}{Lemma}
\newtheorem{corollary}{Corollary}

\newcommand{\hA}{\mathcal{A}}
\newcommand{\hB}{\mathcal{B}}

\newcommand{\hE}{\mathcal{E}}

\newcommand{\hI}{\mathcal{I}} 
\newcommand{\hJ}{\mathcal{J}}

\newcommand{\hM}{\mathcal{M}}

\newcommand{\hR}{\mathcal{R}}
\newcommand{\hS}{\mathcal{S}}

\newcommand{\sfa}{\mathsf{A}}
\newcommand{\sfb}{\mathsf{B}}

\newcommand{\sfe}{\mathsf{E}}
\newcommand{\sff}{\mathsf{F}}
\newcommand{\sfg}{\mathsf{G}}

\newcommand{\sfm}{\mathsf{M}}
\newcommand{\sfp}{\mathsf{P}}
\newcommand{\sfq}{\mathsf{Q}}
\newcommand{\sfz}{\mathsf{Z}}
\newcommand{\sft}{\mathsf{T}}
\newcommand{\sfh}{\mathsf{H}}

\newcommand{\N}{\mathbb N} 
\newcommand{\R}{\mathbb R} 
\newcommand{\Z}{\mathbb Z} 

\newcommand{\fii}{\varphi} 
\newcommand{\Om}{\Omega} 

\newcommand{\hi}{\mathcal{H}} 
\newcommand{\ki}{\mathcal{K}} 
\newcommand{\li}{\mathcal{L}} 

\newcommand{\id}{I} 

\newcommand{\lh}{\mathcal{L(H)}} 
\renewcommand{\th}{\mathcal{T(H)}} 
\newcommand{\sh}{\mathcal{S(H)}} 
\newcommand{\eh}{\mathcal{E(H)}} 
\newcommand{\oh}{\mathcal{O(H)}} 
\newcommand{\ph}{\mathcal{P(H)}} 

\newcommand{\tr}[1]{\mathrm{tr}\left[#1\right]} 
\def\<{\langle} 
\def\>{\rangle} 
\newcommand{\kb}[2]{|#1 \rangle\langle #2|} 
\newcommand{\ip}[2]{\left\langle #1 | #2 \right\rangle} 
\newcommand{\no}[1]{\left\|#1\right\|} 

\newcommand{\Ao}{\mathsf{A}} 

\newcommand{\Eo}{\mathsf{E}} 
\newcommand{\Fo}{\mathsf{F}} 

\newcommand{\Qo}{\mathsf{Q}} 
\newcommand{\No}{\mathsf{N}} 

\newcommand{\bo}[1]{\mathcal{B}\left(#1\right)} 
\newcommand{\br}{\mathcal B(\mathbb R)} 

\begin{document}

\title[Complementary observables]{Complementary observables \\ {\small  in quantum mechanics}}
\author{Jukka Kiukas }
\address{Department of Mathematics, Aberystwyth University, Aberystwyth SY23 3BZ, United Kingdom}
\email{jek20@aber.ac.uk}

\author{Pekka Lahti }
\address{Department of Physics and Astronomy, University of Turku, Turku, Finland}
\email{pekka.lahti@utu.fi}

\author{Juha-Pekka Pellonp\"a\"a }
\address{Department of Physics and Astronomy, University of Turku, Turku, Finland}
\email{juhpello@utu.fi}

\author{Kari Ylinen }
\address{Department of Mathematics and Statistics, University of Turku, Turku, Finland}
\email{ylinen@utu.fi}
\maketitle

\thispagestyle{empty}

\centerline{To the memory of Paul Busch, our friend and colleague}

\vskip .2truecm

\noindent
{\em \small One may view the world with the p-eye and one may view it with  the q-eye
 but if one opens both eyes simultaneously then one gets crazy.}\\
{ \tiny Wolfgang Pauli in a letter to Werner Heisenberg, 19 October 1926.}

\noindent
{\em \small We hope to have demonstrated that one can safely open a pair of complementary `eyes' simultaneously. He who does so may even `see more' than with one eye only. The means of observation being part of the physical world, Nature Herself protects him from seeing too much and at the same time protects Herself from being questioned too closely: quantum reality, as it emerges under physical observation, is intrinsically unsharp. It can be forced to assume sharp contours -- real properties -- by performing repeatable measurements. But sometimes unsharp measurements will be both, less invasive and more informative.}\\
{\tiny Paul Busch {\em et co} in the Epilogue of \cite{OQP}.}

\begin{abstract} We review the notion of \emph{complementarity} of observables in quantum mechanics, as formulated and studied by Paul Busch and his colleagues over the years. In addition, we provide further clarification on the operational meaning of the concept, and present several characterisations of complementarity -- some of which new -- in a unified manner, as a consequence of a basic factorisation lemma for quantum effects. We work out several applications, including the canonical cases of position-momentum, position-energy, number-phase, as well as periodic observables relevant to spatial interferometry. We close the paper with some considerations of complementarity in a noisy setting, focusing especially on the case of convolutions of position and momentum, which was a recurring topic in Paul's work on operational formulation of quantum measurements and central to his philosophy of unsharp reality.
\end{abstract}

\section{Introduction}
Complementarity and  uncertainty are two  key notions of quantum mechanics, and much
 of the scientific work of Paul Busch also dealt with these notions, especially
 with the problem of joint measurability of complementary observables and the relevance of the uncertainty relations to that question. The above quote is a poetic summary of Paul's general thinking on the subject matter -- we dare to say,  even twenty years after its formulation.

In this paper, we study a formulation of the notion of  complementary observables based on an intuitive  idea of  Niels  Bohr, put forward 
 especially in his  1935 paper \cite{Bohr1935}, and strongly advocated by Wolgang Pauli \cite{Pauli1933}, according to which {\em observables are complementary if all the experimental arrangements allowing their unambiguous operational definitions and measurements are mutually exclusive}. Bohr introduced the word complementarity  into the vocabulary of quantum theory  in his  classic Como lecture in 1927
\cite{Bohr1928} aiming to
acquire a consistent interpretation, or, at least, an intuitive understanding of the  then new  quantum mechanical formalism.
 In that paper Bohr  used the term complementarity  several  times in various intuitive contexts never defining it explicitly.
During the years 1927--1962 Bohr published a series of essays in which he strove to develop the idea of complementarity into a definite philosophical viewpoint. 
Most of them are collected in the three volumes entitled {\em Atomic Theory and the Description of Nature}, {\em  Atomic Physics and Human Knowledge}, and {\em Essays 1958--1962 on Atomic Physics and Human Knowledge} originally published in 1934, 1958, and 1963, respectively.
The secondary literature trying to understand Bohr's philosophy is abundant; we mention here only the monographs of   Max Jammer \cite{Jammer1974}, Henry Folse \cite{Folse}, and  Arkady Plotnitsky \cite{Plotnisky}.

Obviously, the experimental arrangements in Bohr's formulation of complementarity cannot be applied together. Therefore, complementary observables cannot be measured jointly.  In this reading, the accompanying 
bold idea of  Werner Heisenberg \cite{WH1927} could be  expressed as follows:\footnote{
For a critical analysis of Heisenberg's ideas on his 1927 paper and their further refinements we refer to the paper \cite{RF2019} of Werner and  Farrelly in this Special Issue. 
} 
complementary observables, like position and momentum, can be defined and measured jointly if sufficient ambiguities are allowed in their definitions. For the necessary defining ambiguities or measurement inaccuracies  $\delta q,\delta p$ for position and momentum Heisenberg gave his famous relation $\delta q\cdot\delta p\sim h$. For an elaboration of these ideas, we refer to the papers \cite{P65}-\cite{P93}  
as well as to the recently initiated
{\em The Quantum Uncertainty Page}  at
  http://paulbusch.wixsite.com/qu-page, with which Paul wanted to serve a large audience interested in the foundational questions of quantum physics.

This paper is structured as follows: In Section \ref{prel} we briefly review the operational formulation of quantum measurement theory as it appeared in most of Paul's work. In Section \ref{order} we collect various characterisations of effect order needed in the subsequent Section \ref{compl}, where we first present an operational definition of complementarity in terms of the lack of joint tests, and then derive a number of general characterisations.  Section \ref{applications} is devoted to applications of the general results to specific observables and their effects, including position-momentum, and interferometric complementarity,  which were central to Paul's work. Finally, in Section \ref{GettingAround} we discuss briefly the topic of our second motto.

\section{Operational formulation of quantum measurement}\label{prel}

The use of a rigorous framework for the quantum measurement theory was undoubtedly one of the main characteristics of Paul's work in general. In the case of complementarity, this is especially important, given the rather philosophical nature of Bohr's original ideas. We now review briefly the relevant concepts.

Throughout the paper we denote by $\hi$  the Hilbert space associated with a physical system under study and by $\lh$  and $\th$ the sets of bounded and trace class  operators on $\hi$.
The concepts of states, observables, and the statistical duality they define form the rudimentary frame of the description of the system: 
a state given as a positive trace one operator $\rho$ acting in $\hi$,  
an observable given as  a normalized positive operator measure $\sfe:\hA\to\lh$, defined on a $\sigma$-algebra  $\hA$ of subsets of a set $\Omega$, 
the probability measure $\hA\ni X\mapsto \sfe_\rho(X)= \tr{\rho\sfe(X)}\in[0,1]$ giving the measurement outcome statistics for the observable $\sfe$ in the state $\rho$.\footnote{We use freely the standard notations and results of quantum theory described in a greater detail, for instance, in the monograph \cite{QM}.}  

Observables are thus identified (and operationally defined) through the totality of their measurement outcome distributions $\sfe_\rho,$ $\rho\in\sh$, with $\sh$ denoting the set of all states of the system.
In addition to this purely statistical level of  description, there are  two deeper levels which take into account 
the  conditional state changes  of the system caused by a measurement on it, or even adopting the most comprehensive level of modeling the interaction and information transfer between the system and the measuring apparatus. Indeed,  each observable $\sfe$ can be realized with a measurement scheme $\hM=(\ki,\sigma,\sfz,U)$, with $\ki$ being the probe Hilbert space, $\sigma$ the initial probe state, $\sfz$ the pointer observable, and $U$ the unitary measurement coupling. If $\hI$ is the instrument defined by $\hM$, then the three levels of the statistical description given by quantum mechanics  get expressed as follows: for any state $\rho$ and for any $X\in\hA$,
\begin{equation}\label{qtm1}
\sfe_\rho(X)=\tr{\rho\sfe(X)}=\tr{\hI(X)(\rho)}=\tr{U(\rho\otimes\sigma)U^*I\otimes\sfz(X)}.
\end{equation}
In fact, any observable $\sfe$ can be identified with an equivalence class of (completely positive) instruments $\hI$ satisfying \eqref{qtm1}, whereas any such instrument $\hI$ can be identified with an equivalence class of  measurements $\hM$ fulfilling \eqref{qtm1}. We wish to emphasize the interpretation of the non-normalized state $\hI(X)(\rho)$ as  a conditional state giving rise to conditional probabilities in the sense that for any other observable $\sff$, with the value space $(\Omega',\hB)$,  the number $\tr{\hI(X)(\rho)\sff(Y)}=\tr{\rho\,\hI(X)^*(\sff(Y))}$ is the probability that a measurement of $\sff$ leads to a result in $Y\in\hB$, given that in  the first performed  $\sfe$-measurement, with the instrument $\hI$, a result in $X\in\hA$  was obtained.

As complementarity represents an extreme case of incompatibility, we also recall at this point the definition of the latter: two (or more) observables  $\sfe_1$ and $\sfe_2$ are compatible or jointly measurable if they have a joint observable, that is, there is an observable $\sfg$
defined on the product $\sigma$-algebra $\hA_1\otimes\hA_2$ of subsets of $\Omega_1\times\Omega_2$ 
 having the two as the marginal observables, that is, for instance, $\sfe_1(X)=\sfg(X\times\Omega_2)$ for all $X\in\hA_1$.

Observables are effect valued measures whereas instruments are operation valued measures. As we will see below, complementarity is defined in terms of the effects constituting the observables, and the order structure plays a central role. Let $\eh$ denote the set of effects (operators $E\in\lh$ with $0\leq E\leq I$)  and $\oh$ the set of operations (completely positive linear maps $\Phi:\th\to\th$ with $0\leq\tr{\Phi(\rho)}\leq 1$ for any state $\rho$). As is obvious from the definitions,  they both are naturally ordered. 
We also recall that
any operation $\Phi\in\oh$ defines an effect $\Phi^*(I)\in\eh$ through  its dual operation $\Phi^*:\lh\to\lh$ 
and any effect $E\in\eh$ is of the form $E=\Phi^*(I)$ for some  $\Phi\in\oh$. Defining two operations equivalent if their effects are the same one gets a bijective correspondence between the effects and the equivalence classes of operations. With a slight abuse of notation, we write $\Phi\in E$ if $\Phi^*(I)=E$ and we say that the operation $\Phi$ implements the effect $E$. Similarly, we write $\hI\in\sfe$ if the instrument $\hI$  defines the observable $\sfe$, that is, for any $X\in\hA$, one has $\sfe(X)=\hI(X)^*(I)$.

\section{On the order structure of the set of effects}\label{order}
Complementarity of observables will be defined and characterised below in terms of order properties of pairs of their effects. This section develops the necessary framework.

\subsection{Square root and other factorisations}
The characterisations of complementarity appearing in this paper are all based on factorising an effect into a product of two contractions. While these results are all elementary and appear in the literature, they have not been systematically applied in the context of complementarity.

For any $E\in\eh$, we let 
$E^{\frac 12}$ be its 
square root, and note that the support space $\hi_E$ of $E$ is
 $$
\mathcal H_E=(\ker E)^\perp=\overline{{\rm ran}\, E}=(\ker E^{\frac 12})^\perp=\overline{{\rm ran}\, E^{\frac 12}}=
\overline{{\rm ran}\, E_0^{\frac 12}}, 
$$
with $E_0$ and $E_0^{\frac 12}$ denoting the restrictions  of $E$ and $E^{\frac 12}$ to $\hi_E$. We let  $P_E$ denote the support projection of $E$, that is, the projection onto the support subspace $\mathcal H_E$. Occasionally, we also let $\sfe^A$ denote the spectral measure of a selfadjoint operator $A$. 

\begin{remark}\label{huomautus1}\rm
We note that the restrictions $E_0$ and $E_0^{\frac 12}$ are bijective onto their ranges. In particular, if $0\in \sigma(E_0)$ then $0\in \sigma_c(E_0)$, and therefore $\sfe^{E_0}(\{0\})=0$, which implies that $x\mapsto x^{-1}$ and $x\mapsto x^{-\frac 12}$ are always measurable and $\sfe^{E_0}$-almost everywhere finite on $[0,1]$. Hence the inverses of the bijections $E_0:\mathcal H_E\to {\rm ran}\, E_0$ and $E_0^{\frac 12}:\mathcal H_E\to {\rm ran}\, E^{\frac 12}_0$ can be constructed via functional calculus, that is,
\begin{align}
{\rm dom}\, E_0^{-\frac 12} &=\left\{\varphi\in \mathcal H_E \,\Big|\, \int_{[0,1]} x^{-1} \sfe^{E_0}_{\varphi,\varphi}(dx)<\infty\right\} = {\rm ran}\, E_0^{\frac 12},\label{domain}\\
\langle \psi| E_0^{-\frac 12}\varphi \rangle &= \int_{[0,1]} x^{-\frac 12}\sfe^{E_0}_{\psi,\varphi}(dx) \,\text{ for all } \psi\in \mathcal H_E,\; \varphi\in {\rm dom}\,E_0^{-\frac 12}\nonumber,
\end{align}
and a similar statement holds for $E_0^{-1}$.\footnote{Here, e.g.\ $\sfe^{E_0}_{\psi,\varphi}$ denotes the complex measure $X\mapsto\ip{\psi}{\sfe^{E_0}(X)\varphi}$.} 
In particular, $E_0^{-\frac 12}$ is selfadjoint on the domain \eqref{domain}, which provides a useful characterisation of the range of $E^{\frac 12}$. Note that by the Hellinger-Toeplitz theorem, ${\rm ran}\, E_0^{\frac 12}=\mathcal H_E$ if and only if $E_0^{-\frac 12}$ is bounded, which is equivalent to the analogous statement for $E_0^{-1}$, and hence further equivalent to $0\notin \sigma_c(E_0)$.
\end{remark}

We now proceed to state two simple lemmas, from which various characterisations of complementarity can conveniently be derived. These lemmas appear essentially in \cite{douglas}; however, as the short and elementary proofs quite effectively illustrate the structure of effects relevant to complementarity, we have included them here. The first one characterises the order relation in terms of the ``splitting'' of an effect into contractions other than the square root.
\begin{lemma}\label{douglas1} Let $\mathcal H$, $\mathcal K$, $\mathcal M$ be Hilbert spaces
 and $K\in \mathcal{L(H,K)}$, $M\in \mathcal{L(H,M)}$ contractions.\footnote{Here, e.g.\ $\mathcal{L(H,K)}$ is the set of bounded operators from $\hi$ to $\ki$.} 
 The following conditions are equivalent:
\begin{itemize}
\item[(i)] $M^*M\leq K^*K$;
\item[(ii)] there exists a 
contraction $C\in \mathcal{L(K,M)}$ such that $M=CK$ and $\ker K^*\subset \ker C$;
\item[(iii)] there exists an 
effect $Q\in \mathcal{E(K)}$ such that $M^*M=K^*QK$ and $\ker K^*\subset \ker Q$.
\end{itemize}
In this case $C$ and $Q$ are unique, $Q=C^*C$, and $\|C\|^2 =\|Q\|= \inf\{ \lambda\in [0,1] \mid M^*M\leq \lambda K^*K\}$.
\end{lemma}
\begin{proof} Assuming (ii), any $Q\in \mathcal{E(K)}$ with $M^*M=K^*QK$ and $\ker K^*\subset \ker Q$ has $K^*C^*CK =M^*M=K^*QK$, so $Q=C^*C$ as both $Q$ and $C$ vanish on $(\overline{{\rm ran}\, K})^\perp=\ker K^*$. Hence (iii) holds. Clearly, (iii) implies (i) as $Q\leq \id$. Assuming (i) we have $\|M\varphi\|^2 \leq \|K\varphi\|^2$ for each $\varphi\in \mathcal H$, so the map $K\varphi \mapsto M\varphi$ from ${\rm ran}\, K$ into ${\rm ran}\, M$ is well defined and extends to a contraction $C\in \mathcal {L(K)}$, which is unique if required to vanish on $(\overline{{\rm ran}\, K})^\perp =\ker K^*$, so (ii) holds. Hence (i)-(iii) are equivalent and $Q=C^*C$ when they hold. In this case $r=\inf\{ \lambda\in[0,1] \mid M^*M\leq \lambda K^*K\}\leq \|Q\|=\|C\|^2$ as $M^*M=K^*QK\leq \|Q\| K^*K$, and if $\lambda\in [0,1]$ has $M^*M\leq \lambda K^*K$, then $\|M\varphi\|^2\leq \lambda \|K\varphi\|^2$ for all $\varphi\in \mathcal H$, so $\|C\|^2 \leq \lambda$ by the construction of $C$. Hence $\|C\|^2=r$.
\end{proof}

The second lemma relates order to the inclusion of the ranges of the contractions appearing in the first lemma.
\begin{lemma}\label{douglas2} Let $K$ and $M$ be as in the above lemma. The following are equivalent:
\begin{itemize}
\item[(i)] $M^*M\leq \lambda K^*K$ for some $\lambda\geq 0$;
\item[(ii)] ${\rm ran}\, M^*\subset {\rm ran}\, K^*$.
\end{itemize}
\end{lemma}
\begin{proof} Clearly, (i) implies (ii) by Lemma \ref{douglas1}. Furthermore, the restriction of $K^*$ to $(\ker K^*)^\perp=\overline{{\rm ran}\, K}$ is bijective onto $\mathcal D={\rm ran}\, K^*$ with inverse $(K^*)^{-1}: \mathcal D\to \overline{{\rm ran}\, K}$ densely defined and closed in $\overline{\mathcal D}=\ker K$, since a sequence $(\varphi_n)$ in $ \mathcal D$ for which $\lim_n \varphi_n=\varphi\in \ker K$ and $\lim_n (K^*)^{-1}\varphi_n=\psi$ also has $\lim_n \varphi_n =\lim_n K^* (K^*)^{-1}\varphi_n = K^*\psi$ as $K^*$ is bounded, so $\varphi = K^*\psi \in {\rm ran}\, K^*=\mathcal D$ and $(K^*)^{-1}\varphi = \psi$. If (ii) holds then ${\rm ran}\, M^* \subset \mathcal D$ so $R=(K^*)^{-1}M^*$ is defined on all of $\mathcal M$, and closed as $(K^*)^{-1}$ is closed and $M^*$ bounded. Hence $R\in \mathcal{L(M,K)}$ by the closed graph theorem, so $K^*R=M^*$ and hence $M^*M=K^*RR^*K\leq \|R\|^2K^*K$, proving (i).
\end{proof}

\subsection{Range and order}
We now show how two different characterisations of effect order follow from the above factorisation lemmas.

Our first application is the following proposition. In order to state it we recall some relevant terminology: the one-dimensional projections $P[\fii]=\kb{\fii}{\fii}$, $\fii\in\hi$, $\no{\fii}=1$,  
are the atoms of the projection lattice $\ph$ and any $P\in\ph$ is the join (the least upper bound) of all the atoms contained in it. Also, the meet of any two projections exists both in $\ph$ and in $\eh$ and is clearly the projection onto the intersection of the ranges of the two projections. 
Though there are no atoms in $\eh$, it is convenient to call any rank-1 effect 
$\kb{\fii}{\fii},\fii\in\hi$,  with $0\ne\no{\fii}\leq 1$,  
a  weak atom. 
 According to \cite[Corollary 3]{BG1999}  each effect is the join  of all the weak atoms contained in it.
On the other hand, the weak atoms contained in an effect  $E$ are characterised  by \cite[Theorem 3]{BG1999}:
\begin{proposition}\label{Paul_Stan}
Let $E$  be an effect and $\kb\fii\fii$ a weak atom. Then
$$
\exists \lambda >0:\; \lambda \kb\fii\fii
\leq E \; \Longleftrightarrow \; \fii\in {\rm ran}\,E^{\frac 12}.
$$
Moreover, then $\sup\{\lambda\geq 0 \mid \lambda |\varphi\rangle\langle \varphi| \leq E \} = \big\|E_0^{-\frac 12}\varphi\big\|^{-2}.$
 \end{proposition}
 \begin{proof} Follows immediately from Lemmas \ref{douglas2} and \ref{douglas1}.
 \end{proof}

The second application is \emph{dilation}: every effect $E$ can be dilated to a projection $P$ on a larger Hilbert space, as $E=J^*PJ$ where $J$ is an isometry. This is an instance of the well-known Naimark dilation theorem, and clearly a particular case of the above factorisation. Hence we can easily derive the following result:

\begin{lemma}\label{dilat} Let $E\in \mathcal E(\mathcal H)$ be an effect, $\psi\in \mathcal H$ with $\|\psi\|\leq 1$, and $E=J^*PJ$, $J\in\mathcal L(\hi,\ki)$, a Naimark dilation of $E$ into a projection $P\in \mathcal{L(K)}$. The following conditions are equivalent:
\begin{itemize}
\item[(i)] $|\psi\rangle\langle \psi| \leq E$;
\item[(ii)] there is an $\eta\in \overline{{\rm ran}\, PJ}$, $\|\eta\|\le1$, such that $\psi = J^*\eta$.
\end{itemize}
\end{lemma}
\begin{proof} We use Lemma \ref{douglas1} with $M=\langle \psi|:\mathcal H\to \mathbb C$ and $K=PJ$, so the adjoint of $C:\mathcal K\to \mathbb C$ has one-dimensional range ${\rm ran} \,C^*\subset \overline {{\rm ran}\, PJ}\subset {\rm ran}\, P$ (where the second inclusion is due to ${\rm ran}\, P$ being closed). Taking $\eta\in {\rm ran} \,C^*$ the lemma gives $\psi = J^*P\eta = J^*\eta$.
\end{proof}

\begin{remark} More generally, condition (iii) of Lemma \ref{douglas1} yields the following statement: if $E = J^*PJ$ is a dilation of an $E\in \mathcal{E(H)}$, then $A\leq E$ for $A\in \mathcal{E(H)}$ if and only if $A = J^*QPJ$ for a $Q\in \mathcal{E(K)}$ commuting with $P$. (Commutativity follows since ${\rm ran}\, Q\subset \overline{{\rm ran}\, PJ}\subset {\rm ran}\, P$.) This is a simple special case of the Radon-Nikodym theorem for completely positive maps; see e.g.\ \cite{Raginsky}, which could therefore also be used to derive the lemma. Since we do not need the general statement, the above elementary proof is justified. 
\end{remark}

\subsection{Bounding the support projection}
In many relevant cases (such as position and momentum; see below), the effect is constructed via functional calculus from some existing selfadjoint operator. While every effect can be written in this form, 
 the setting becomes interesting when the function has a nontrivial structure -- the cases of smearing of a sharp observable with a Markov kernel or a convolution with a probability measure fall into this category. The following Lemma is relevant in this context:
\begin{lemma}\label{necessary1}
If $E=f(A)=\int f\,d\sfa$ for some spectral measure $\sfa:\,\bo\R\to\lh$, with the selfadjoint operator $A=\int x\,\sfa(dx)$,  and a Borel measurable function $f:\,\R\to[0,1]$, then $P_E\leq \sfa({\rm supp}(f))$, but equality does not hold in general.\footnote{We use the notation $\bo T$ for the Borel $\sigma$-algebra of any topological space $T$.}
\end{lemma}
\begin{proof}
We have $\ip{\fii}{E \fii}= \int f\, d\sfa_{\fii,\fii} \leq\int_{{\rm supp}(f)} d\sfa_{\fii,\fii}= \ip{\fii}{\sfa({\rm supp}( f))\fii}$
for all $\varphi\in \mathcal H$, so if $\varphi$ is orthogonal to $\sfa({\rm supp}( f))(\mathcal H)$ then $\varphi\in \ker E = \mathcal H_E^\perp = P_E^\perp(\mathcal H)$. This proves the first statement. 

Let $\sfa$ be $\sfq$, the canonical spectral measure on $L^2(\R)$, 
 $C$  a Cantor set with  positive measure, and $f(x)$ the minimum of 1 and the distance of $x$ from $C$.
Then $f$ is a continuous nonnegative function $f$ such that ${\rm supp}(f)=\mathbb R$ and thus 
$\sfa({\rm supp}(f))=I$.
However, the characteristic function $\chi_C\in \ker f(A)\setminus\{0\}$ so that $P_{f(A)}\ne I$. 
\end{proof}

Note that for an arbitrary effect $E$, an application of the lemma with $\sfa =\sfe^E$ and $f(x) =x$, $x\in[0,1]$, $f(x)=0$, $x\in\R\setminus[0,1]$, is consistent with the  fact $P_E = \sfe^E([0,1])$, but does not provide any more information -- as noted above, interesting cases arise with nontrivial functions.

\subsection{Effect order and pure operations}\label{alaotsikko26}
The order of effects has no direct relation to the order of the operations implementing them. Clearly, if  $A\leq E$ then for any fixed state $\sigma$, the operations $\Phi^A_\sigma(\rho)=\tr{\rho A}\sigma$ and $\Phi^E_\sigma(\rho)=\tr{\rho E}\sigma$, defining $A$ and $E$, respectively, are  also ordered $\Phi^A_\sigma\leq\Phi^E_\sigma$. However, these operations are maximally noisy in the sense that the normalised
post-measurement state (with or without conditioning on a specific outcome) is always the fixed  state $\sigma$, which is unrelated to the effects $A$ and $E$ under consideration. In the other extreme, there are the L\"uders  operations associated with  the ideal, first kind, repeatable measurements of discrete sharp observables, the operations of the form $\Phi_L^P(\rho)= P\rho P$, $P\in\ph$. For discrete unsharp observables  their counterpart are the generalised L\"uders operations,  $\Phi^E_L(\rho)=E^{\frac 12}\rho E^{\frac 12}$, extensively studied also by Paul, see, e.g., \cite{PCompendium}.
These are a special case of the pure operations $\rho\mapsto K\rho K^*$, $K\in\lh$, defining an effect $E=K^*K$.
Since any operation can be written as a sequence of pure operations, one often argues that pure operations have the least amount
of classical noise.
 In any case, they are specific to the observables and hence imprint some
information on the measurement to the post-measurement state. In this sense they form the opposite of the  trivial operations  $\Phi_\sigma^E$.

\begin{proposition}\label{proppure} Let $\Lambda$ and $\Phi$ be two pure operations with $\Lambda^*(I)=A$ and $\Phi^*(I)=E$. The following are equivalent:
\begin{itemize}
\item[(i)] $A\leq E$;
\item[(ii)] there exists a pure operation $\Psi$ such that $\Lambda = \Psi\circ\Phi$.
\end{itemize}
\end{proposition}
\begin{proof} Writing $\Lambda = M(\,\cdot\,)M^*$ and $\Phi = K(\,\cdot\,)K^*$ we can apply Lemma \ref{douglas1}; the contraction $C$ featuring in condition (ii) determines $\Psi =C(\,\cdot\,)C^*$.
\end{proof}

\subsection{Common lower bounds for a pair of effects}
In preparation for the discussion on complementarity in the next section, we now consider joint lower bounds for pairs of effects.

For any two effects $E,\,F\in\eh$, we let ${\rm l.b.}\{E,F\}=\{A\in\eh\,|\, A\leq E, \,A\leq F\}$ denote the set of their common lower bounds, and similarly 
 for the operations $\Phi,\,\Psi\in\oh$. In neither case does their meet, the greatest lower bound, typically exist.\footnote{A characterization of the existence of the infimum  of effects is established in \cite{DU_etal2006}. In particular, 
if one of the effects is a projection, then their meet  exists \cite{Moreland_Gudder1999}.} This is the case even if  $E$ and $F$ are compatible, in which case they are of the form   $E=\sfe(X)$ and $F=\sfe(Y)$ for some  observable $\sfe$ and thus  $\sfe(X\cap Y)$ is a common lower bound of them; still, ${\rm inf}\{E,F\}=E\land F$ need not exist in $\eh$.
The characterization of the set ${\rm l.b.}\{E,F\}$, and especially the case  ${\rm l.b.}\{E,F\}=\{0\}$  is the key issue of this study. Proposition \ref{Paul_Stan} gives the following:\footnote{An earlier version of  Proposition \ref{Paul_Stan} together with the equivalence of  \eqref{lowerbounds} and  \eqref{Paulin_ehto} was obtained in \cite{P2}.}
\begin{corollary}\label{GenJauch}
For any two effects $E,$ $F\in\eh$ the following  conditions are equivalent:
\begin{eqnarray}
{\rm l.b.}\{E,F\}&\ne& \{0\};\label{lowerbounds}\\
{\rm ran}\,E^{\frac 12}\cap{\rm ran}\,F^{\frac 12}&\ne& \{0\}.\label{Paulin_ehto}
\end{eqnarray}
\end{corollary}

A direct study of \eqref{lowerbounds} and \eqref{Paulin_ehto} may, in general, be challenging, since the range of an effect need not be closed. In fact, while the condition $\hi_E\cap\hi_F=\{0\}$ (i.e.\ $P_E\land P_F=0$) clearly implies $E\land F=0$, the converse need not hold even under additional constraints, as will be demonstrated below by Proposition \ref{Jukan_esimerkki1}. In some cases, the following necessary condition, which follows directly from Lemma \ref{necessary1}, is more tractable.

\begin{lemma}\label{necessary}
Let $E=f(A)$, $F=g(B)$ where $A$ and $B$ are selfadjoint operators given by (real) spectral measures $\sfa$ and $\sfb$, and $f:\R\to[0,1]$ and $g:\R\to[0,1]$ are measurable. If $\sfa({\rm supp}(f))\land\sfb({\rm supp}(g))=0$, then ${\rm l.b.}\{E,F\}=\{0\}$.
\end{lemma}

\section{Complementary observables}\label{compl}

Intuitively, observables are complementary if  the experimental arrangements allowing their unambiguous definitions are  mutually exclusive.
With the full machinery of quantum mechanics,
 one may formalise the concept of `experimental arrangement unambiguously defining an observable' 
using either  the measurement schemes  defining the observable, the instruments implementing it, or just the observable, itself.
We follow \cite{BugL1980} and \cite{P3} to express the idea of  `mutual exclusiveness of experimental arrangements' in terms of the order structure of the sets of effects and operations.

\subsection{A test of a binary observable}
Consider any two  observables $\sfe$ and $\sff$ with the outcome $\sigma$-algebras $\hA$ and $\hB$.
If  the set ${\rm l.b.}\{\sfe(X),\sff(Y)\}\ne \{0\}$ for some $X$ and $Y$,
then for any (nozero) effect $A$ which is below $\sfe(X)$ and $\sff(Y)$, the yes-outcome  1 of a yes-no measurement of the dichotomic observable
$\sfa$, with $\sfa(1)=A,$ $\sfa(0)=I-A=A^\perp$, 
gives probabilistic information on both of the effects $\sfe(X)$ and $\sff(Y)$. If the effects $\sfe(X)$ and $\sff(Y)$ are disjoint, that is, $\sfe(X)\land\sff(Y)=0$,
equivalently, $\Phi\land\Psi=0$ for any $\Phi\in\sfe(X),\;\Psi  \in\sff(Y)$, 
 then no such measurements exist.  We elaborate next the operational context  of this idea a bit further.

Let $\Qo$ be a binary observable (a yes/no question) with outcomes $1$ and $0$, where $1$ denotes the yes-answer. Suppose we have another binary $1/0$--observable $\Ao$ such that
\begin{itemize}
\item[(1)] $\Ao$ can be measured jointly with $\Qo$;
\item[(2)] the $1$-outcome of $\Ao$ serves as a definite indicator for the $1$-outcome of $\Qo$, that is, the latter occurs with certainty given that the former is $1$, for any state of the system;
\item[(3)] the ``indicator'' outcome $1$ of $\Ao$ has nonzero probability at least for some state.
\end{itemize}
We call such an observable \emph{a test for $\Qo$}.

For any observable $\Eo$ and $X\in \mathcal A$, we let $\Qo_{\Eo,X}$ denote the binary coarse-graining of $\Eo$ corresponding to the question of whether the outcome lies in $X$, that is, $\Qo_{\Eo,X}(1) = \Eo(X)$ and $\Qo_{\Eo,X}(0) = \id -\Eo(X)$. The following simple observation follows readily from the definition. 

\begin{proposition}\label{char1} Let $\Qo$ and $\Ao$ be binary $1/0$--observables. The following are equivalent:
\begin{itemize}
\item[(i)] $\Ao$ is a test for $\Qo$;
\item[(ii)] there exists an observable $\Eo$ with outcome $\sigma$-algebra $\mathcal A$, and sets $Y,\,X\in \mathcal A$, $Y\subset X$, such that $\Eo(Y)\neq 0$ and $\Qo_{\Eo,X}=\Qo$ and $\Qo_{\Eo,Y}=\Ao$; 
\item[(iii)] $\Ao(1)\leq \Qo(1)$;
\item[(iv)] there exists a contraction $C$ such that $\sqrt{\Ao(1)}=\sqrt{\Qo(1)}\,C$.
\end{itemize}
\end{proposition}
\begin{proof} (i)$\Leftrightarrow$(iii): If (i) holds then by the joint measurability of $\sfa$ and $\Qo$  there are four effects $G_{00}$, $G_{10}$, $G_{01}$, $G_{11}$ summing to identity, such that $G_{10}+G_{11}=\Qo(1)$ and $G_{01}+G_{11}=\Ao(1)$, where the first and second indices refer to outcomes of $\Qo$ and $\Ao$, respectively. This implies that $G_{11}\leq \Ao(1)\leq P$ where $P$ is the projection onto the support subspace of $\Ao(1)$. Then $P\neq 0$ by the condition (3), and $G_{11}$ and $\Ao(1)$ are determined by states with support in $P$. Since the conditional probability condition (2) reads ${\rm tr}[G_{11}\rho]/{\rm tr}[\Ao(1)\rho]=1$ for any such state, we must have $G_{11}=\Ao(1)$, so that $\Qo(1)-\Ao(1)=G_{10}\geq 0$, that is, (iii) is true. Conversely, if (iii) holds then $\Ao$ and $\Qo$ have the joint observable $G_{11}=\Ao(1)$, $G_{10}=\Qo(1)-\Ao(1)$, $G_{01}=0$, $G_{00}=\id -\Qo(1)$, and the conditional probability condition is satisfied with similar remarks on the support. (ii)$\Leftrightarrow$(iii): If (iii) holds the three-outcome observable $\Eo =\{\Ao(1), \Qo(1)-\Ao(1), \Qo(0)\}$ satisfies the requirements of (ii), and if (ii) holds then $\Ao(1) = \Eo(Y)\leq \Eo(X)= \Qo(1)$ so (iii) holds as well. The equivalence of (iii) and (iv) follows from Lemma \ref{douglas1}.
\end{proof}

\subsection{The definition of complementarity as a lack of joint tests}

We are now ready to state the definition of complementarity, and give a basic characterisation based on the results of the preceding section.

\begin{definition}\rm 
Let $\Eo$ and $\Fo$ be two observables with outcome $\sigma$-algebras $\mathcal A$ and $\mathcal B$. Given $X\in \mathcal A$ and $Y\in \mathcal B$, the observables  $\Eo$ and $\Fo$ are \emph{$(X,Y)$-complementary} if $\Qo_{\Eo,X}$ and $\Qo_{\Fo,Y}$ have no common tests. Given $\mathcal A_0\subset \mathcal A$ and $\hB_0\subset\hB$, we say that $\Eo$ and $\Fo$ are $(\mathcal A_0,\mathcal B_0)$-\emph{complementary} or, briefly,
\emph{complementary} (if $\hA_0$ and $\hB_0$ are clear from context\footnote{See discussion on the choice of $\hA_0$ and $\hB_0$ below.}),
if they are $(X,Y)$-complementary for each $X\in\mathcal A_0$ and $Y\in \mathcal B_0$. 
\end{definition}

\begin{theorem}\label{seuraus1} Observables $\sfe:\hA\to\lh$ and $\sff:\hB\to\lh$ are  $(\hA_0,\hB_0)$-complementary, if and only if one of the equivalent conditions below hold:
\begin{itemize}
\item[(i)] for any $X\in\hA_0,$ $Y\in\hB_0$,  
the effects $\sfe(X)$ and $\sff(Y)$ are disjoint, that is, $\sfe(X)\land\sff(Y)=0$;
\item[(ii)] ${\rm ran}\sqrt{\sfe(X)}\cap{\rm ran}\sqrt{\sff(Y)}=\{0\}$ for all $X\in\hA_0,$ $Y\in\hB_0$;
\item[(iii)] any two instruments $\hI\in\sfe$ and $\hJ\in\sff$ are mutually exclusive with respect to $\hA_0$ and $\hB_0$, that is, $\hI(X)\land\hJ(Y)=0$ for all $X\in\hA_0,$ $Y\in\hB_0$; 
\item[(iv)] for any pure operations $\Phi\in \sfe(X), \; \Psi\in \sff(Y)$, there exist no pure operations $\Lambda_1,\;\Lambda_2$ such that $\Lambda_1\circ\Phi=\Lambda_2\circ\Psi$.
\end{itemize}
\end{theorem}

Regarding the choice of $\hA_0$ and $\hB_0$, the naive choice $\mathcal A_0=\mathcal A$ and $\mathcal B_0=\mathcal B$ obviously leads to a trivial notion, as the identity operator has a joint lower bound with any effect. Merely excluding the identity would still be too strong a requirement, restricting complementarity essentially only to dichotomic observables. However, for an unambiguous definition of  an observable $\sfe:\hA\to\lh$ one does not need all its effects -- indeed, using polarisation and the 
Carath\'eodory extension theorem we  see that 
it suffices to specify the effects $\sfe(X),$ $X\in\hR$, for some semiring $\hR\subset\hA$ which generates $\hA$ and  covers $\Omega$ in the sense of a 
countable union (of sets that can, moreover, be required to be disjoint, as one can easily show).\footnote{We recall that $\hS\subset 2^\Omega$ is a semiring if $\emptyset\in\hS$,  for all $X,Y\in\hS$, $X\cap Y\in\hS$, and $X\setminus Y$ is the union of finite number of disjoint sets  belonging to $\hS$.} In concrete examples this allows one to choose the sets $\hA_0$ and $\hB_0$ such that they contain such generating covering semirings. With this restriction, complementarity becomes a special case of quantum incompatibility:
\begin{proposition} Complementary observables have no joint measurements.
\end{proposition}
\begin{proof}
If $\sfm$ is a joint measurement of $\sfe$ and $\sff$ (with the semirings $\hR$ and $\hS$) and $(X_i)\subset\hR\subset\hA_0,$ $(Y_j)\subset\hS\subset\hB_0$ countable disjoint covers for $\Omega$ and $\Omega'$, then $I=\sfm(\Omega\times\Omega')=\sum_{i,j}\sfm(X_i\times Y_j)$, implying that $\sfm(X_i\times Y_j)\neq 0$ for some $(i,j)$, providing a joint lower bound for $\sfe(X_i)$ and $\sff(Y_j)$.
\end{proof}

We now discuss briefly the choice of $\mathcal A_0$ and $\mathcal B_0$ in two basic cases:
\subsubsection{Continuous case}\label{cont}
For real observables absolutely continuous with respect to the Lebesgue measure, one could choose $\Omega={\rm supp}(\sfe)\subset\R$ and 
$\hA_0\subset \hA=\br\cap{\rm supp}(\sfe)$ to consist of {\it bounded} Borel sets $X$ for which $\Omega\setminus X$ has nonzero Lebesgue measure. This choice excludes the identity, and satisfies the generating semiring condition. However, we could equally well include {\it all} Borel sets with the above restriction regarding the measure, leading to a different notion of complementarity. In fact, the canonical position-momentum pair is complementary in the former but not in the latter sense, as we will discuss later on.

\subsubsection{Discrete case}\label{discr}

If $\sfe$ is discrete (and nontrivial), the outcome set is essentially $\Omega=\{x_1,x_2,\ldots\}$ where $\Eo(x)\ne 0$ for all $x\in \Omega$, and $\sum_{x\in \Omega} \Eo(x)=I$.
In this case, $\hA=2^\Omega$, and the obvious generating semiring is $\big\{\{x\}\,\big|\, x\in\Omega\big\}$. Clearly, there are (at least) two natural choices: (1) $\hA_0$ consists of all finite proper subsets of $\Omega$, and (2) $\hA_0=\big\{\{x\}\,\big|\, x\in\Omega\big\}$.

The first choice will be relevant for the examples in Section \ref{applications}. Regarding the second one, Proposition \ref{proppure} yields an interesting characterisation in terms of conditional post-measurement states.
In fact, complementarity of $\sfe$ and $\sff$ excludes the possibility that these could be further post-processed into the same final conditional state:

\begin{proposition} Let $\sfe$ and $\sff$ be discrete observables with outcome sets $\Omega$ and $\Omega'$, and let $\mathcal A_0=\{\{x\}\mid x\in \Omega\}$ and $\mathcal B_0=\{\{y\}\mid y\in \Omega'\}$. Then $\sfe$ and $\sff$ are $(\hA_0,\hB_0)$-complementary if and only if 
their generalized  L\"uders instruments $\hI^L$ and $\hJ^L$ do not satisfy $\Psi_x\circ\hI^L(\{x\})=\Phi_y\circ\hJ^L(\{y\})$ for any pair $x\in \Omega$, $y\in \Omega'$, and 
 any pure operations $\Psi_x$ and $\Phi_y$.
\end{proposition}

\subsection{Complementarity in terms of dilations}\label{JPjakso}

In this section we characterise complementarity using Naimark dilations of the observables; this method will then be further refined in applications. We consider two observables 
$\Eo:\,\hA\to\lh$ and
$\Fo:\,\hB\to\lh$ with the value spaces $(\Om,\hA)$ and $(\Om',\hB)$, and 
let $(\hi_\oplus,\Qo,J)$, resp.\ $(\hi'_\oplus,\Qo',K)$,  be a minimal diagonal Naimark dilation of $\Eo$, resp.\ $\Fo$ (see, for instance, \cite[Sec.\ 8.6]{QM}).
For instance, $\hi_\oplus$ is a direct integral Hilbert space, $\Qo:\,\hA\to\li(\hi_\oplus)$ its canonical spectral measure, and $J:\,\hi\to\hi_\oplus$  an isometry such that $\Eo(X)=J^*\Qo(X)J$ for all $X\in\hA$. Lemma \ref{dilat} now yields the following characterisation.

\begin{proposition}\label{theorem1}
$\Eo$ and $\Fo$ are $(\hA_0,\hB_0)$-complementary if and only if for each $X\in \hA_0$, $Y\in \hB_0$,
\begin{equation}\label{condition}
J^*\eta=K^*\eta',
\qquad \eta\in\overline{{\rm ran}[\Qo(X)J]}, 
\quad \eta'\in\overline{{\rm ran}[\Qo'(Y)K]},
\end{equation}
implies $\eta=0$ (or, equivalently, $J^*\Qo(X)\eta=0$). 
\end{proposition}

\begin{remark}\label{remarkJP}
We note the following relevant facts:
\begin{enumerate}
\item If $\eta\in\overline{{\rm ran}[\Qo(X)J]}$ and $\eta'\in\overline{{\rm ran}[\Qo'(Y)K]}$ then
 $\Qo(\Omega\setminus X)\eta=0$ and $\Qo'(\Omega'\setminus Y)\eta'=0$.
\item \label{4}
If, say, $\Eo$ is projection valued, then $J$ is unitary, and \eqref{condition} reads 
$\eta=F^*\eta'$ with $\eta\in{{\rm ran}\,\Qo(X)}$ and
$F^*F=I_{\hi_\oplus}$, i.e.\ $F=KJ^*$ is an isometry, $\hi\cong\hi_\oplus$, $\dim\hi_\oplus\le\dim\hi'_\oplus$.
If also $\Fo$ is projective, then $\hi\cong\hi_\oplus\cong\hi'_\oplus$; moreover
$\eta\in{{\rm ran}\,\Qo(X)}$ and $\eta'\in{{\rm ran}\,\Qo'(Y)}$ so that 
$\Eo(X)\wedge\Fo(Y)=0$ is equivalent to ${{\rm ran}\,\Qo(X)}\cap F^*\big({{\rm ran}\,\Qo'(Y)}\big)=\{0\}$.
\end{enumerate}
\end{remark}

\subsection{Other formulations of complementarity}\label{Ludwig}
There is a stronger form of complementarity advanced, for instance, in \cite{OQP,P3,P38}. 
Accordingly,  two observables $\sfe$ and $\sff$ could be called {\em strongly complementary} if 
$$
\sfe(X)\land\sff(Y)=\sfe(X)\land\sff(Y)^\perp=\sfe(X)^\perp\land\sff(Y)=0
$$
for any $X\in\hA_0,\; Y\in\hB_0$. 
Some of the natural pairs of observables are known to be  complementary but not strongly complementary (examples below) which is why we consider here the weaker formulation as the generic notion.

The notion of extreme incompatibility may also be used  to express complementarity.  Here we quote three  versions of extreme incompatibility.

First,  let $\hE$ be the collection of trivial effects $\lambda I$, $0\leq \lambda\leq 1$, and let $\sft=p I$ be a trivial observable,  defined by  a probability measure $p$.
According to Ludwig  \cite[D 3.3, p.\ 154]{Ludwig1983}
two observables $\sfe$ and $\sff$ are L-{\em  complementary}, L for Ludwig,  if neither of them is a trivial observable
 and for each observable $\sfe'$ it follows that 
 $$\sfe(\hA)\cap\sfe'(\hA')\subset\hE\quad {\rm or}\quad \sff(\hB)\cap\sfe'(\hA')\subset\hE.
$$
Clearly, this is an extreme case of incompatibility. In fact, if $\sfe_1$ and $\sfe_2$ are  L-complementary  then they cannot have any mutually commuting effects in their ranges. Indeed, if $E\in\sfe_1(\hA_1)$ and $F\in\sfe_2(\hA_2)$ are mutually commuting, then any joint observable $\sfe$ of the dichotomic observables $\{0,E,E^\perp,I\}$ and $\{0,F,F^\perp,I\}$ would contradict their L-complementarity.

As is  well known, any two observables  $\sfe_1$ and $\sfe_2$ can be made compatible by adding trivial noise in the form
$\widetilde\sfe_1=\lambda\sfe_1+(1-\lambda)\sft_1$ and $\widetilde\sfe_2=\mu\sfe_2+(1-\mu)\sft_2$, where  $0\leq\lambda,\,\mu\leq 1,$ and $\sft_1,$ $\sft_2$ are trivial observables, see, for instance, \cite{P75}. Let $J(\sfe_1,\sfe_2)$  denote the set of pairs $(\lambda,\mu)\in[0,1]\times[0,1]$ for which there exist  $(\sft_1,\sft_2)$ such that 
$\widetilde\sfe_1$ and $\widetilde\sfe_2$ are compatible. Then $\Delta\subset J(\sfe_1,\sfe_2)$, where  $\Delta=\{(\lambda,\mu)\,|\, \lambda+\mu\leq 1\}$.  Clearly, if $(1,1)\in J(\sfe_1,\sfe_2)$, then $\sfe_1$ and $\sfe_2$ are compatible. In the other extreme,
  $J(\sfe_1,\sfe_2)=\Delta$ and the observables may be called maximally incompatible.
 If $j(\sfe_1,\sfe_2)$ denotes the supremum of the set of the numbers $0\leq\lambda\leq 1$ such that $(\lambda,\lambda)\in J(\sfe_1,\sfe_2)$, then $\sfe_1$ and $\sfe_2$ are maximally incompatible if and only if $\lambda=\frac 12$ \cite{TJ_etal_2014}.

Finally, there is a slightly different form of maximal incompatibility  especially useful for sharp observables. The degree of commutativity of two projections $P,$ $R\in\ph$ can be desribed in terms of their commutativity projection
$$
{\rm com}(P,R)=(P\land R)\lor (P\land R^\perp)\lor (P^\perp\land R)\lor (P^\perp\land R^\perp),
$$
the range of which consists exactly of the vectors $\fii\in\hi$ for which $PR\fii=RP\fii$. Clearly, $0\leq {\rm com}(P,R)\leq I$, the extreme cases indicating total noncommutativity and commutativity, respectively. One can further refine the totally noncommutative case in terms of the spectrum of the effect $PRP$ (or, equivalently, $RPR$) -- the case where the spectrum is the whole $[0,1]$ represents maximal incompatibility in the sense of robustness against arbitrarily biased binary noise -- for instance, position and momentum projections corresponding to half-lines fall into this category \cite{noiserob}.

For any sharp observables $\sfa$ and $\sfb$ we have
${\rm com}(\sfa,\sfb)=\bigwedge_{X\in\hA,\,Y\in\hB}\,{\rm com}(\sfa(X),\sfb(Y))$,
and it is known \cite{Kari1985} that a unit vector $\fii$ is in the range of this projection
exactly when there is a probability measure $\mu:\hA\otimes\hB\to[0,1]$ such that
$\mu(X\times Y)=\ip{\fii}{\sfa(X)\sfb(Y)\fii} = \ip{\fii}{\sfa(X)\land\sfb(Y)\fii}$.
Clearly, observables $\sfa,\;\sfb$ are compatible if this projection is the identity $I$. On the other hand, if  ${\rm com}(\sfa,\sfb)=0$  the observables can be called  totally incompatible.

\section{Examples of complementarity}\label{applications}

\subsection{The canonical case: position and momentum in $L^2(\R)$}\label{QPsection}
Position and momentum are the  prototype pair of complementary observables. They were also central to the work of Paul Busch. Therefore, we start with a brief discussion of the well-known results in this setting, and then proceed to derive a few new results on the complementarity of the relevant unsharp localisation effects, focusing on the (previously less studied) case where one of them is periodic. 

Let $\sfq$ and $\sfp$ be the canonical position and momentum observables in  $\hi=L^2(\R)$, and  $F$ be the Fourier-Plancherel operator. We let $Q$ and $P$ denote the corresponding selfadjoint position and momentum operators, so that $P=F^*QF$.
\subsubsection{Complementarity of $\sfq$ and $\sfp$}
As is  well known\footnote{Equations \ref{Lenard0} express the basic  fact that the support of the Fourier transform of a compactly supported function is the whole  $\R$. The result \eqref{Lenard} is derived in \cite{Lenard}.},
for bounded $X,\,Y\in\br$,
\begin{eqnarray}
&&\sfq(X)\land\sfp(Y)=\sfq(X)\land\sfp(\R\setminus Y)=\sfq(\R\setminus X)\land\sfp(Y)=0, \label{Lenard0}\\
&&\sfq(\R\setminus X)\land\sfp(\R\setminus Y)\ne 0.\label{Lenard}
\end{eqnarray}
Since bounded sets in $\br$ (even bounded intervals) form a generating semiring $\hR$ that covers $\R$ (in the sense of countable union) we conclude that $\sfq$ and $\sfp$ are not only complementary but even strongly complementary\footnote{Note that ${\rm supp}(\sfq)=\R={\rm supp}(\sfp)$.}  observables, if we choose $\hA_0=\hB_0=\hR$.
Moreover, it can be shown \cite{TJ_etal_2014} that they are maximally incompatible in the sense of joint measurability region discussed in Sec. \ref{Ludwig}. Finally, since
$$
{\rm com}(\sfq,\sfp)\leq \bigwedge_{X,Y\in \hR}{\rm com}(\sfq(X),\sfp(Y))
=\bigwedge_{X,Y\in \hR}\sfq(\R\setminus X)\land \sfp(\R\setminus Y)=0,
$$
the canonical pair $(\sfq,\sfp)$  is also totally incompatible in the sense of trivial commutativity domain. 
On the other hand, they are not L-complementary and also not $(\hA_0,\hB_0)$-complementary if we include countable unions of sets of $\hR$ in $\hA_0$ and $\hB_0$. Indeed, for any periodic sets $X+a,$ $Y+b$, with minimal positive periods $a,b$ satisfying $\frac{2\pi}{ab}\in\N$, one has $\sfq(X)\sfp(Y)=\sfp(Y)\sfq(X)$ (see, for instance, \cite[Theorem 15.2]{QM}).

This means that they have jointly measurable coarse grainings, not only of the binary form associated with the above type of sets, 
but of the form $X\mapsto\sfq^f(X)=\sfq(f^{-1}(X))$ and $Y\mapsto\sfp^g(Y)=\sfp(g^{-1}(Y))$, where $f$ and $g$ are essentially bounded periodic Borel functions with minimal positive periods $a,b$ satisfying $\frac{2\pi}{ab}\in\N$ (see, e.g.\ \cite[Theorem 15.2]{QM}).

\subsubsection{Complementarity of derived effects}

In addition to the sharp observables $\sfq$ and $\sfp$, it is also natural to study effects derived from them as first moments, assuming that $0\leq f\leq 1$ and $0\leq g\leq 1$, that is, $E=f(Q)$ and $F=g(P)$. In the philosophy of Paul Busch, these correspond to \emph{unsharp} properties,\footnote{The concept of unsharp property has gradually refined in the work of Paul Busch, with the first explicit definition being given in \cite[Definition 4]{Paul84}.}
which for suitable functions will approximate the  {\em sharp} properties $\sfq(X)$ and $\sfp(Y)$. There are then different cases of complementarity: on the one hand, if the functions have compact support, their support projections are complementary by the above discussion, and hence the effects themselves remain complementary due to Lemma \ref{necessary}. On the other hand, in the periodic case specified above, the effects are commutative and hence non-complementary, even compatible. As a basic example of the latter case, we consider
\begin{align}\label{haversincom}
E &= f_0(Q) , & F &= g_0(P),
\end{align}
with ``haversin'' and ``havercos'' functions $f_0(x)=\frac 12(1-\cos x)$ and $g_0(p)=\frac 12(1+\cos ((2\pi)^{-1}p))$. These effects represent unsharp periodic localisation: for instance, $E$ assigns small probabilities to states concentrated around points $2\pi n$, $n\in \mathbb Z$, and large ones to states around $2\pi n+\pi$. 

It is now interesting to consider the intermediate case -- we retain the periodic effect $E$ on the $Q$-side but compress $F$ into one periodicity interval; in the haversin example, this yields
\begin{align}\label{haversin}
E &= f_0(Q) , & F &  = \sfp\big(\textstyle{[-\frac 12,\frac 12]}\big)(P)g_0(P)\sfp\big(\textstyle{[-\frac 12,\frac 12]}\big).
\end{align}
Note that this can be understood operationally as conditioning on the $P$-measurement.

Since $P_E=\id$, the support projections are not disjoint in this case, and we need to look into the structure of the ranges more closely. It turns out that this pair is in fact complementary -- we now proceed to prove this result in a more general case, which applies both in the context of multislit interferometry and ``generalised Jauch theorem'' considered later in this paper.

Accordingly, let $E=f(Q)$ and  $F=g(P)$, where $f,g:\mathbb R\to[0,1]$ are measurable functions, and denote by $\mathcal Z_f=f^{-1}(\{0\})$ the zero set of $f$ and $\mathcal S_f =\mathbb R\setminus \mathcal Z_f$ its complement set. By  Remark \ref{huomautus1} the support subspaces and critical domains are now given by
\begin{align*}
\mathcal H_{E}&=L^2(\mathcal S_f), & \mathcal H_F&=F^{-1}L^2(\mathcal S_g),\\
{\rm ran}\, E^{\frac 12} &= \left\{\phi\in \mathcal H_E\Big| \int \frac{|\phi(x)|^2}{f(x)}dx <\infty\right\},& {\rm ran}\, F^{\frac 12} &= \left\{\phi\in \mathcal H_F\Big| \int \frac{|\widehat\phi(p)|^2}{g(p)}dp <\infty\right\}.
\end{align*}
\begin{remark} Before proceeding, a couple of subtleties are worth pointing out.

(a) Since $E^{\frac 12}$ acts as $\big(E^{\frac 12}\psi\big)(x)= \sqrt{f(x)} \psi(x)$, one might think that each continuous function $\phi\in {\rm ran}\, E^{\frac 12}$ must vanish whenever $f$ does, assuming $f$ is also continuous. Of course, this is not the case; 
for instance, if $f(x) = \sqrt{|x|}$ for $|x|\leq 1$ and 1 for $|x|>1$, then any $\phi\in L^2(\mathbb R)$ which is constant on $[-1,1]$, belongs to ${\rm ran} \,E^{\frac 12}$.

(b) Another issue is related to the support subspace: By Lemma \ref{necessary1}, $\mathcal H_E= L^2(\mathcal S_f)\subset L^2({\rm supp} (f))$ where ${\rm supp} \,f=\overline{ \mathcal S_f}$, but the inclusion may be strict. Hence there could be cases where $\mathcal H_E\cap \mathcal H_F=\{0\}$ but $\sfq({\rm supp}(f))\land\sfp({\rm supp}(g))= 0$ does not hold.
\end{remark}

\begin{proposition}\label{Jukan_esimerkki1} Suppose that $g\neq 0$ is measurable and compactly supported, and fix any $R>0$ such that ${\rm supp}(g)\subset [-R,R]$. Then assume that $f$ is continuous with zero set $\mathcal Z_f=\{n\pi/R\mid n\in \mathbb Z\}$ and $x\mapsto f(x)^{-1}$ not integrable over any neighbourhood of any $x_0\in \mathcal Z_f$. Then $E$ and $F$ are complementary but $\mathcal H_E\cap \mathcal H_F \neq \{0\}$.
\end{proposition}

\begin{proof} Since $g\neq 0$, $\mathcal S_g\subset [-R,R]$ has nonzero measure, and hence we have $\{0\}\neq \mathcal H_F\subset F^{-1}L^2[-R,R]$. We also have $\mathcal H_E=L^2(S_f)=L^2(\mathbb R\setminus \mathcal Z_f) = L^2(\mathbb R)=L^2({\rm supp}(f))$. In particular, the intersection $\mathcal H_E\cap\mathcal H_F=\mathcal H_F$ is nontrivial.

Let now $\phi\in {\rm ran}\, E^{\frac 12}\cap {\rm ran }\, F^{\frac 12}$, and note first that $\phi\in {\rm ran }\, F^{\frac 12}\subset \mathcal H_F\subset F^{-1}L^2[-R,R]$ implies $\widehat \phi(p)=0$ for $|p|>R$. In particular, $\phi$ is the inverse Fourier transform of an integrable function (as $L^2[-R,R]\subset L^1[-R,R]$), and hence continuous. Therefore, if $\phi(x_0)\neq 0$ for some $x_0\in \mathcal Z_1$ we could find $\epsilon,\delta>0$ for which $\int_{x_0-\epsilon}^{x_0+\epsilon} f(x)^{-1}|\phi(x)|^2dx\geq \delta \int_{x_0-\epsilon}^{x_0+\epsilon} f(x)^{-1} dx$, which would contradict $\phi\in {\rm ran}\, E^{\frac 12}$ as the second integral is infinite by assumption. Hence we must have $\phi(x)=0$ for all $x\in \mathcal Z_f$, that is, $\phi(n\pi/R)=0$ for all $n\in \mathbb Z$. Next note that the restriction $\widehat \phi\in L^2[-R,R]$ implies that $\widehat \phi = \sum_{n\in \mathbb Z} \langle \psi_n|\widehat \phi\rangle \psi_n$ where $\psi_n(p) = \frac{1}{\sqrt{2\pi}}\chi_{[-R,R]}(p)e^{-inp\pi /R}$ forms an orthonormal basis of the subspace $L^2[-R,R]$. But here $\langle \psi_n|\widehat \phi\rangle= \frac{1}{\sqrt {2\pi}} \int_{-R}^R e^{in\pi p/R} \widehat{\phi}(p)dp = (F^{-1}\widehat \phi)(n\pi/R)=\phi(n\pi/R)$. (This is just the \emph{sampling theorem} from signal analysis, see e.g.\ \cite[p. 230]{Folland}, but we need the above calculation to get the constants right with our definition of $F$.) Since $\phi(n \pi/R)=0$ for all $n\in \mathbb Z$ we have $\widehat\phi =0$, and hence also $\phi=0$. We have shown that ${\rm ran}\, E^{\frac 12}\cap {\rm ran }\, F^{\frac 12}=\{0\}$, that is, $E$ and $F$ are complementary.
\end{proof}

As an example, the complementarity of the pair \eqref{haversin} follows directly from Prop. \ref{Jukan_esimerkki1} (where now $R=1/2$), since the non-integrability condition is satisfied as $f_0(x)\sim x^2$ near zero.

\begin{remark} Note that this example (and also the general context of Proposition \ref{Jukan_esimerkki1}) corresponds to a boundary case where the effect $E=f(Q)$ does not have a kernel but the spectrum still reaches zero, thereby allowing a possibility for complementarity. The case where $\inf f>0$ is uninteresting as ${\rm ran}\,E^{\frac 12}=L^2(\mathbb R)$ and the effect is not complementary with any other effect.
\end{remark}

\subsubsection{Complementarity and (informational) compleness}
We close this section with a brief discussion on another aspect of $QP$-complementarity which was historically significant. The idea of complementarity of observables typically includes also the idea of their equal importance for the full description of the system. Perhaps, it was in this sense that Pauli \cite{Pauli1933} posed the question if the position and momentum distributions suffice to determine the state of the system. It was soon demonstrated by Valentine Bargmann (as reported in \cite{Reichenbach}) that this is not the case.\footnote{See, for instance,  \cite{CH} or \cite{ACT2005} for other explicit examples. }
The informational incompleteness of the complementary pair $(\sfq,\sfp)$ suggests that some of the complementary information is lacking.
This leads one to ask if there is a third observable $\sfh$, say energy,  which is  complementary to $\sfq$ and $\sfp$, and 
which would complete $\sfq$ and $\sfp$ to an informationally complete triple $(\sfq,\sfp,\sfh)$.
Example \ref{JPesimerkki1} 
shows that if the spectrum of the energy is purely discrete, then the position-energy and momentum-energy pairs are complementary, too.
If the energy operator $H$ is of the form $H=\frac 1{2m} P^2+V(Q)$, with $V(Q)$ bounded and positive
it is known that the probability distributions $\sfq_\rho,\,\sfp_\rho,\,\sfh_\rho$ do not  suffice to determine the state $\rho$ \cite{CH,Pavi}.
Not knowing the general answer to the posed question, we recall that
if $\sfq_\theta=U_\theta\sfq U_\theta$, with $U_\theta=e^{i\theta H}$, 
is  a quadrature observable, then, not only the pair $\sfq$ and $\sfp=\sfq_{\frac{\pi}2}$, but, in fact, any pair $(\sfq,\sfq_\theta)$, $\theta\notin\{0,\pi\}$, is complementary \cite{JPP2010}. Moreover,  any family of the pairwise complementary observables $\{\sfq_\theta\,|\, \theta\in S\}$, with a dense set  $S\subset[0,2\pi)$,
is informationally complete \cite{KLP2008}.

\subsection{Complementarity of continuous-discrete pairs}\label{JPjakso}

We now focus on pairs $\sfe$ and $\sff$ where $\sfe$ is absolutely continuous with respect to the Lebesgue measure (as in Sec.\ \ref{discr}) and $\sff$ discrete (as in Sec.\ \ref{cont}). Let $\Omega$ and $\Omega'$ denote the respective outcome sets. Throughout this subsection, $\hA_0$ is the family of Borel subsets of $\Omega$ whose complement has nonzero Lebesgue measure, and $\hB_0$ the family of finite proper subsets of $\Omega'$ (i.e.\ the first choice in Sec.\  \ref{discr}). 

\subsubsection{The case of sharp $\sfe$}\label{JPesimerkki1} Here we assume that $\Eo$ is a rank-1 sharp observable, so that $\Eo(X)=J^*\Qo(X)J$ for all $X\in \mathcal A$, where $\Qo$ is the canonical spectral measure on $\hi_\oplus=L^2(\Omega)$ (with $\Omega\subset \mathbb R$ having nonzero Lebesgue measure), and $J$ is unitary. For $\Fo$ assume that $m_y={\rm rank}\, \Fo(y)<\infty$ for each $y\in \Omega'$, and write $\Fo(y)=\sum_{k=1}^{m_y}\kb{f_{yk}}{f_{yk}}$ where $\{f_{yk}\}_{k=1}^{m_y}$ is linearly independent. 
\begin{proposition}[Polynomial method] Suppose that
there is a (measurable) weight function $w:\,\Omega\to(0,\infty)$ such that $Jf_{yk}$ is a polynomial multiplied by $w$ for each $y\in \Omega'$ and $k=1,\ldots,m_y$. Then $\sfe$ and $\sff$ are 
complementary.
\end{proposition}
\begin{proof}
Define $\Qo'(y)=\sum_{k=1}^{m_y}\kb{\fii_{yk}}{\fii_{yk}}$, where $\{\fii_{yk}\}$ is an orthonormal basis of a Hilbert space $\hi_\oplus'$, and set 
$K=\sum_{y\in \Omega}\sum_{k=1}^{m_y}\kb{\fii_{yk}}{f_{yk}}$. Then $\Fo(y)=K^*\Qo'(y)K$ for all $y\in \Omega'$ so $K$ and $\Qo'$ form a (minimal) dilation of $\Fo$ on $\hi_\oplus'$.
The complementarity condition \eqref{condition} of Prop.\ \ref{theorem1} for $Y\in \hB_0$ now reads
$J^*\eta=K^*\eta'$ where $\Qo(\Omega\setminus X)\eta=0$ and $\Qo'(\Omega'\setminus Y)\eta' =0$. Hence $J^*\eta=\sum_{y\in Y}\sum_{k=1}^{m_y}\<\fii_{yk}|\eta'\>f_{yk}$, so 
$\eta = (JJ^*)\eta=\sum_{y\in Y}\sum_{k=1}^{m_y}c'_{yk}Jf_{yk}$ with $c'_{yk}=\<\fii_{yk}|\eta'\>$ as $J$ is unitary.
Hence by assumption $\eta$ is a polynomial times $w$, and as such either has a finite number of zeros or is identically zero.
But the former case is impossible, as $\eta(x)=0$ for (almost) all $x\in \Omega\setminus X$ and $\Omega\setminus X$ has nonzero measure. 
Hence $\eta=0$, and an application of Proposition \ref{theorem1} completes the proof.
\end{proof}
In particular, the canonical spectral measure $\Eo$ of any interval $\Omega\subset\R$ has many complementary discrete observables.
In the case of a bounded interval we can choose $\Omega=[-1,1]$ by a simple transformation and the basis consists of Jacobi polynomials or trigonometric polynomials (times a weight factor). The interval $\Omega=[0,\infty)$ gives associated Laguerre polynomials and $\R$ Hermite polynomials. If we let $\Fo$ be the number observable associated with the chosen polynomial basis we see that the canonical spectral measure and the number are complementary observables. Especially, the number (or energy) and the position (or momentum) of the harmonic oscillator form complementary pairs; in this case $\Omega=\R$ and the basis consists of Hermite polynomials multiplied by the Gaussian weight.
 More generally, if $\Fo$ is the energy $\frac12 P^2+V$ where the potential $V$ is such that the energy spectrum is discrete and its eigenvectors are (e.g.\ trigonometric) polynomials (with a weight) then position and energy are complementary (e.g. a particle in a box).

\subsubsection{Circular position and number}\label{JP_example_torus} 
Let $\Om=[0,2\pi)$ and fix a $Z\subset\Z$. Let $\{|n\>\}_{n\in Z}$ be an orthonormal basis of $\hi$, $\Fo=\No$, the number, i.e.\
$\Fo(n)=\kb{n}{n}$, and $\Eo(X)=J^*\Qo(X)J$ where $J=\sum_{n\in Z}\kb{e_n}{n}$, $e_n(\theta)=(2\pi)^{-\frac12}e^{-in\theta}$, and $\Qo$ is the position observable of the circle. Note that this case includes, in particular, the case of periodic position and momentum (with $Z=\mathbb Z$), as well as the canonical phase \cite[p.\ 372]{QM} and number (or the canonical time \cite[p.\ 402]{QM} and energy) of the harmonic oscillator ($Z=\mathbb N$).
 Complementarity in the former case was studied in \cite{LY1987}, while the latter case was treated recently in \cite{LPS2017}. 

Since $\Eo$ is a spectral measure only when $Z=\mathbb Z$, other cases (including number-phase) are not covered by the polynomial method. Nevertheless, the following result holds:
\begin{proposition}\label{proptorus} If $\mathbb Z\setminus Z$ is bounded from below or above,
then $\Eo$ and $\Fo$ are complementary.
\end{proposition}
\begin{proof}
Recall that the condition \eqref{condition} of Prop.\ \ref{theorem1} reads $J^*\eta = \eta' $ where $\Qo\big([0,2\pi)\setminus X\big)\eta =0$ and $\Fo(Z\setminus Y)\eta'=0$, where $X\in \hA_0$ and $Y\in \hB_0$. (Note that now $K=\id$ since $\Fo$ is a spectral measure.) Hence $J^*\eta = \eta'= \sum_{n\in Y} \langle n|\eta'\rangle |n\rangle$, so
$JJ^*\eta=\sum_{n\in Y}\<n|\eta'\>e_n$
where $Y\subset Z$ is finite (and $Y\ne Z$). But $\eta= \sum_{n\in \mathbb Z} \<e_n|\eta\>e_n$ so $JJ^*\eta= \sum_{n\in Z} \<e_n|\eta\>e_n$, which implies that $\<e_m|\eta\>=0$ for all $m\in Z\setminus Y$, that is,
$\eta=\sum_{m\in Y\cup(\Z\setminus Z)}\<e_m|\eta\>e_m$. Suppose first that $\Z\setminus Z$ is bounded from above. Then $Y\cup(\Z\setminus Z)$ is also bounded from above and so $\eta=\sum_{m=-\infty}^{m_Y}\<e_m|\eta\>e_m$ for some $m_Y\in\Z$, that is, $\eta$ is (up to a phase) a Hardy function on the circle. But $\Qo\big([0,2\pi)\setminus X\big)\eta =0$, that is, $\eta(x)=0$ for (almost all) $x\in [0,2\pi)\setminus X$, where $[0,2\pi)\setminus X$ has positive measure. Hence $\eta=0$ (see e.g.\ \cite{helson}). An even easier reasoning applies for the `below' case.
\end{proof}

\subsection{Multislit interferometry}\label{mslit}
Being a classic application of complementarity, spatial interferometry was also studied by Paul Busch until recently \cite{P76,P85}. The following idealised setting (which however quite well approximates typical experimental situations, see the cited papers) neatly illustrates several aspects of complementarity studied above. Consider an infinite periodic aperture mask given by the periodic set $A=\cup_{n\in \mathbb Z} (X+n)$ where $X\subset [-1/2,1/2]$ describes a single slit. Then in the usual Fraunhofer approximation, the position measurement at a detector screen placed at a large distance will correspond to a momentum measurement in the coordinates of the aperture, with a typical interference pattern typically exhibiting periodic structure with the ``inverse'' period $2\pi$.

\subsubsection{Complementarity of ``which way'' and interference measurements}

Consider the following decomposition of $Q$ and $P$ into a sum of periodic part 
and the remainder:
\begin{align*}
Q&=Q_{\rm mod} + Q_{\rm d}, & P&=P_{\rm mod} + P_{\rm d},
\end{align*}
So here $Q_{\rm mod}$ is ``$Q$ modulo $1$'', coinciding with $Q-n\id$ in $L^2([-1/2,1,2]+n)$, and $Q_{\rm d}$ is the discretised position with eigenvalues $n\in \mathbb Z$ labelling the slit index, with eigenprojections $\sfq(X+n)$. The decomposition of $P$ is similar in the momentum space, but with period $2\pi$.

Note that here the canonical complementary pair $(Q,P)$ is decomposed by separating out the commuting part: indeed, $[Q_{\rm mod}, P_{\rm mod}]=0$ as these operators are periodic functions of $Q$ and $P$ of the form discussed at the beginning of Section \ref{QPsection}. In the momentum space, $P_{\rm mod}$ captures the periodic structure of the interference pattern while $Q_{\rm d}$ is the ``which way'' measurement giving the information on the slit the particle ``has passed through''. Hence they form an appropriate pair of interferometric observables, as originally suggested in \cite{Aharonov} and further studied in \cite{P76, P85}.

\begin{proposition} The pair $(\sfq_{\rm d}, \sfp_{\rm mod})$ is complementary (where $\hA_0$ and $\hB_0$ are as in the preceding subsection).
\end{proposition}
\begin{proof} As noted in \cite{P85}, one can easily check that the unitary groups generated by these operators satisfy the Weyl relations for the phase space $\mathbb T\times \mathbb Z$, and therefore by the Stone-von Neumann-Mackey theorem $(Q_{\rm d}, P_{\rm mod})$ is a direct sum of copies of the associated canonical pair whose complementarity was proved in Prop.\ \ref{proptorus}. Hence the claim follows, as it is clear from Theorem \ref{seuraus1} that complementarity is preserved when taking direct sums.
\end{proof}

\subsubsection{From commutativity to complementarity}
It is furthermore interesting to consider the transition from commutativity to complementarity due to the inclusion of the ``which path'' information. This can conveniently be done in the level of effects: consider first measuring the commutative effects $E=f(P_{\rm mod})$ and $F=g(Q_{\rm mod})$, where $f$ is taken to be continuous and vanishing exactly where $P_{\rm mod}$ does, i.e.\ at points $n2\pi$. For definiteness, we can take $f=f_0$ and $g=g_0$ as in \eqref{haversincom} considered in Section \ref{QPsection}. Here we can regard $E=f(P_{\rm mod})$ as an unsharp yes/no measurement regarding whether the value of $P_{\rm mod}$ is zero or not.

Now the compression of $F$ into $F'=\sfq([-\frac 12, \frac 12])F\sfq([-\frac 12, \frac 12])$  
 by the projection $\sfq([-\frac 12, \frac 12])$ 
onto the slit at the interval $[-\frac 12,\frac 12]$ can be interpreted as conditioning on the ``which path'' information that the particle ``passed through'' this specific slit, leading (up to Fourier-transform) to the pair \eqref{haversin}, which is indeed complementary even though the corresponding support subspaces have a nontrivial intersection.

\section{Complementarity and noise}\label{GettingAround}
Complementarity is an extreme form of incompatibility, and as such, one could expect that  
it would be unstable against the addition of noise. We first make some remark on the general aspects of this phenomenon, in the level of pairs of generic effects, and then proceed to consider the case of unsharp position and momentum observables with convolution type noise. Complementarity in the latter context was considered by Paul Busch in his  1984 paper, entitled ``On joint lower bounds of position and momentum observables in quantum mechanics'' \cite{P2}.

\subsection{Breaking complementarity of effects by noise}

The following simple results show how a small perturbation immediately regularises any effect $E$ so that its inverse becomes bounded.
\begin{proposition}\label{ex1} For any $E\in \mathcal{E(H)}$ and $\lambda,\,p\in (0,1)$, define two modified effects
\begin{align*}
E_{\lambda,p} &= (1-\lambda)E+\lambda p\id, & E_p=p(\id-E) +(1-p)E,
\end{align*}
corresponding to classical noise addition and convolution (see proposition below). Then ${\rm ran}\, E_{\lambda,p}^\frac 12={\rm ran}\, E_{p}^\frac 12=\mathcal H$, regardless of how small $p$ and $\lambda$ are.
\end{proposition}
\begin{proof} Consider first $E_{\lambda, p}$. Since $\sigma((1-\lambda) E)\subset [0,1]$, it follows that $-\lambda p\in (-1,0)$ is in the resolvent set of $(1-\lambda) E$, and hence $E_{\lambda,p} = (1-\lambda) E -(-\lambda p) \id$ has a bounded inverse, that is, $0\notin \sigma(E_{\lambda,p})$. Regarding $E_p$, note that if $p\leq \frac 12$ then $\sigma((1-2p)E)\subset [0,1-2p]$ so $-p\in (-1,0)$ is in the resolvent set of $(1-2p)E$ and hence $E_p=(1-2p)E+p\id$ has a bounded inverse. Similarly, if $p>\frac 12$ then $\sigma((1-2p)E)\subset [1-2p,0]$ so $-p\in (-1,1-2p)$ is again in the resolvent set. Hence, in both cases $0\notin \sigma(E_p)$.
\end{proof}

\begin{corollary} Let $\Qo$ be a binary observable and $p\in (0,1)$. Define the convolution $$\Qo_p = \{p\Qo(0) +(1-p)\Qo(1), (1-p)\Qo(0)+p\Qo(1)\}$$ with outcomes $1$ and $0$, respectively. Then ${\rm ran} \sqrt {\Qo_p(0)} ={\rm ran} \sqrt {\Qo_p(1)}=\mathcal H$, and hence $(\Qo_p,\Qo')$ is not $(i,j)$-complementary for any binary observable $\Qo'$ and any $i,j=0,1$.
\end{corollary}

In order to make a slightly more definitive statement, we let $\mathcal C \subset \mathcal{E(H)}\times \mathcal{E(H)}$ denote the set of pairs of 
complementary effects.
An immediate observation regarding the stability of complementarity is that in the finite-dimensional case complementary effects cannot have full rank, and hence complementarity can be destroyed by arbitrary small perturbations by trivial observables proportional to identity. The same result holds in the infinite-dimensional case; a precise formulation can be stated as follows:
\begin{proposition}\label{prop8}
 Equip $\mathcal{E(H)}$ with any induced vector space topology of $\mathcal{L(H)}$  and $\mathcal{E(H)}\times \mathcal{E(H)}$ with the corresponding cartesian product topology. Then the subset $\mathcal C$ has empty interior.
\end{proposition}
\begin{proof}  Let $E, F\in \mathcal E(\hi)$. Then for any $\lambda>0$ the effects $E_\lambda = (1-\lambda) E+\lambda \id$ and $F_\lambda = (1-\lambda) F+\lambda \id$ have ${\rm ran}\, \sqrt {E_\lambda} = {\rm ran}\, \sqrt {F_\lambda} =\mathcal H$ by Prop. \ref{ex1}, so $(E_\lambda, F_\lambda)\notin \mathcal C$. Since $[0,1]\ni \lambda \mapsto (E_\lambda, F_\lambda)\in \mathcal{E(H)}\times \mathcal{E(H)}$ is continuous, the claim follows.
\end{proof}
These results demonstrate that complementarity, like other forms of extreme incompatibility,  is not stable in \emph{arbitrary small} perturbations even in the infinite-dimensional case, and is instantly destroyed in mixtures with trivial observables. This is very different from incompatibility, in general,  which is typically preserved until some nontrivial noise threshold also in the finite-dimensional case.
 This is most evident in the case of qubit effects and observables. Indeed, any two (different) qubit effects are complementary if and only if they are of rank-1, that is,  weak atoms.  Moreover, two qubit observables are complementary if and only if they are sharp and clearly any two  sharp qubit observables are complementary. By contrast, for any two qubit effects $E=\frac 12(e_0I+\vec e\cdot\vec\sigma)$ and  $F=\frac 12(f_0I+\vec f\cdot\vec\sigma)$, and thus for the corresponding dichotomic observables $\sfe$ and $\sff$, 
 their compatibility can be  expressed in the form of a single inequality -- in fact, $E$ and $F$ are compatible exactly when
\begin{align*}
\ip{E}{E^\perp}\ip{F}{F^\perp}  - \,  &\sqrt{\ip{E}{E}\ip{F}{F}\ip{E^\perp}{E^\perp}\ip{F^\perp}{F^\perp} }\leq \ip{E}{F^\perp}\ip{E^\perp}{F} + \ip{E}{F}\ip{E^\perp}{F^\perp}
\end{align*}
where, for instance, $\ip{E}{F}=\frac 14(e_0f_0-\vec e\cdot\vec f)$ \cite[Theorem 3]{P70}; see \cite{P7} for the original proof in a special case. 
For an extensive study of the compatible approximators of the complementary sharp qubit (spin) observables, see \cite{P91,P93,P75,TJ_etal_2014, noiserob}.

\subsection{Generalised Jauch theorem}\label{GJT}
 We return to the context of the paper \cite{P2}, where Paul proved
`generalized Jauch theorem'  to answer the question how much unsharpness in the form of convolutions needs to be introduced into position and momentum in order  to break their complementarity \eqref{Lenard0}. 
 According to this theorem, our Corollary \ref{GenJauch},
for any pair of  unsharp position and momentum observables  $\mu*\sfq$ and  $\nu*\sfp$ and   for any of their value  sets $X$ and $Y$,
$$
{\rm l.b.}\{(\mu*\sfq)(X),(\nu*\sfp)(Y)\}\ne \{0\}
$$
if and only if 
$$
{\rm ran}\sqrt{(\mu*\sfq)(X)}\cap{\rm ran}\sqrt{(\mu*\sfq)(X)}\ne \{0\}.
$$
By Lemma \ref{necessary1} the support spaces $\hi_E$ and $\hi_F$ of the   effects  $E=(\mu*\sfq)(X)$ and $F=(\nu*\sfp)(Y)$ are contained in the subspaces
$\sfq({\rm supp}(\chi_X*\mu))(\hi)$ and $\sfp({\rm supp}(\chi_Y*\nu))(\hi)$, respectively, so that there are two obvious necessary conditions for the noncomplementarity of these effects:
\begin{align*}
&\sfq({\rm supp}(\chi_X*\mu))(\hi)\cap\sfp({\rm supp}(\chi_Y*\nu))(\hi)\ne \{0\},&\\
&\hi_E\cap\hi_F\ne \{0\}.&
\end{align*}
Clearly, if 
$\sfq({\rm supp}(\chi_X*\mu))\land\sfp({\rm supp}(\chi_Y*\nu))= 0$, and thus also $P_E\land P_F=0$, 
 then the effects $(\mu*\sfq)(X)$ and $(\nu*\sfp)(Y)$ remain complementary. It will be shown below that these implications cannot be reversed.

\begin{remark} From Lemma \ref{necessary1} we know that e.g. $\mathcal H_E$ could in principle be strictly contained in $\sfq({\rm supp}(\chi_X*\mu))(\hi)$, but we do not construct an example here -- in what follows we consider the case where $\mathcal H_E=\sfq({\rm supp}(\chi_X*\mu))(\hi)$.
\end{remark}
 
Before proceeding to the relevant result, we recall the following observation on the support of the involved convolutions: since  ${\rm supp}(\chi_X*\mu)\subset \overline{\overline X+{\rm supp}(\mu)}$ and ${\rm supp}(\chi_Y*\nu)\subset \overline{\overline Y+{\rm supp}(\nu)}$; from this we may conclude, along  with \cite{P2}, that if the measures $\mu$ and $\nu$ have bounded supports then the unsharp observable $\mu*\sfq$ and $\nu*\sfp$ are still complementary. The following remark explores the support question in a slightly more general context.

\begin{remark}
\rm
In this remark we let $G$ denote a (not necessarily abelian) locally compact group, using multiplicative notation in general but additive notation
in the abelian case (i.e., in (b)).
Let $M(G)$ be the space of regular complex Borel measures on $G$. We regard it as equipped
with the convolution product $(\mu,\nu)\mapsto \mu * \nu$ as in  \cite{HeRo63}.

(a) If $G$ is compact, then for any probability measures $\mu,\,\nu\in M(G)$ the support ${\rm supp}(\mu * \nu)$
equals ${\rm supp}(\mu){\rm supp}(\nu)$,  the set of the products $xy$ with $x\in{\rm supp (\mu)}$, $y\in{\rm supp} (\nu)$
 \cite[p.\ 925]{Wendel54}.

(b) If $G$ is not assumed to be compact, the claim in (a) need not hold, even if $G$ is abelian. To see this, let $G={\mathbb R}^2$,
let $\mu$ be the probability measure on $G$ supported by the $x$-axis and defined by the $N(0,1)$  Gaussian density on the $x$-axis, and let
$\nu=\sum_{n=1}^{\infty}2^{-n}\delta_{(n,n^{-1})}$. The support of the convolution $\mu * \nu$ contains the $x$-axis which, however, is not contained
in ${\rm supp}(\mu) + {\rm supp}(\nu)$.

(c) We claim that ${\rm supp}(\mu * \nu)$ is contained in the closure of ${\rm supp}(\mu){\rm supp}(\nu)$ for any probability measures $\mu,\,\nu\in M(G)$. We
only needed this result above in the abelian case, but the proof does not require commutativity. It is enough to show that
$\int_{G}f \, d(\mu * \nu)=0$ whenever $f:\,G\to [0,1]$ is a continuous function with compact support contained in the open complement of the closure of ${\rm supp}(\mu){\rm supp}(\nu)$ (see \cite[p.\ 123]{HeRo63}). For such an $f$,
$\int_{G}f\,d(\mu * \nu)=\int_G d\mu(x)\int_Gf(xy)d\nu(y)=\int_{{\rm supp}(\mu)} d\mu(x)\int_{{\rm supp}(\nu)}f(xy)d\nu(y)=0$, since $f(xy)=0$ whenever $x\in{\rm supp}(\mu)$ and $y\in{\rm supp}(\nu)$.
\end{remark}

The following application of Proposition \ref{Jukan_esimerkki1} shows that the necessary condition given above in terms of the supports of the convolving measures is not sufficient. 

\begin{proposition}\label{Jukan_esimerkki2} For any bounded intervals $X,\,Y\subset \mathbb R$ with lengths $d_X,$ $d_Y$ satisfying $d_Xd_Y\leq \pi/2$, there exist probability density functions $f,g$ with finite variance, such that the effects $f*\chi_{X}(Q)$ and $g*\chi_Y(P)$ are complementary, but $\sfq({\rm supp}(f*\chi_X))\land \sfp({\rm supp}(g*\chi_Y))\neq 0$.
\end{proposition}
\begin{proof}
Choose
$f(x)= \sum_{n\in \mathbb Z} p_n \chi_{[-\frac 12,\frac 12]}(x-2n)$ 
where $p_n>0$ are such that $\sum_{n\in \mathbb Z} p_n =1$ and $\sum_{n\in \mathbb Z} n^2p_n<\infty$. (More general functions could be chosen.) Then $f$ is a probability density function with mean zero and finite variance.
Moreover, $f$ vanishes exactly on the periodic set $\cup_{n\in \mathbb Z} (2n+[\frac 12,\frac 32])$. Hence, if we take $X=[\frac12,\frac 32]$ then  $h_1(x) = (f*\chi_X)(x) =\sum_{n\in \mathbb Z} p_n f_0(x-2n)$,
where $f_0 = \chi_{[-\frac 12,\frac 12]}*\chi_{[\frac12,\frac 32]}$, that is, $f_0(x) =x$ for $x\in[0,1]$, $f_0(x)=2-x$ for $x\in[1,2]$ and zero otherwise. Now $h_1$ vanishes precisely in $2\mathbb Z$, that is, $\mathcal Z_1=2\mathbb Z$, and we also have ${\rm supp}(h_1)=\mathbb R$, so $\mathcal H_E = L^2({\rm supp}(h_1))=L^2(\mathbb R)$ in this case. Moreover, for each $n\in \mathbb Z$ we have $h_1(x) = |x-2n|$ whenever $|x-2n|<1$ and so $1/h_1(x)$ is not integrable over any open interval containing $2n$. Hence $h_1$ satisfies the conditions of Proposition \ref{Jukan_esimerkki1}. Now if we let $g=2\pi^{-1}\chi_Y$ where $Y=[-\pi/4, \pi/4]$, then $g$ is a probability density with mean zero and finite variance, and if we take $h_2=g*\chi_{Y}$, then ${\rm supp}(h_2)=[-\pi/2,\pi/2]$ and hence $\mathcal H_F = L^2(S_2) = L^2({\rm supp}(h_2))= L^2([-\pi/2,\pi/2])$. Then Proposition \ref{Jukan_esimerkki1} applies with $R=\pi/2$.

We can generalise this construction: if $X$ is any bounded interval of length $d_X$, a density function $f$ can be constructed as above to have zero set $\cup_{n\in \mathbb Z}(2nd_X+X)$, so that $f*\chi_X$ is zero exactly at the equidistant points $2d_X\mathbb Z$. Hence if $Y$ is a bounded interval centred at $0$ with length $d_Y$ and $g$ is the uniform distribution on $Y$ then $g * \chi_Y$ has support in $[-d_Y,d_Y]$ and hence any $\widehat \phi$ supported there is determined by the values $\phi(n\pi/R)$ if $R\geq d_Y$. If we adjust $R=\pi/(2d_X)$ to match this with $2nd_X$ and require $\phi\in {\rm ran} (f*\chi_X(Q))^{\frac 12}$ then we must have $\phi(n\pi/R)$ and hence $\phi=0$ as above. The required choice of $R$ is possible if $d_Xd_Y\leq \pi/2$.
\end{proof}

\begin{remark} The first part of the above proof shows that the mean and variance of $f$ and $g$ can be chosen to be zero if e.g.\ $X=[\frac 12,\frac 32]$ and $Y=[-\pi/4, \pi/4]$.
\end{remark}

\section{Summary}
We have reviewed and reconsidered the notion of complementarity, as advanced by Paul Busch and his colleagues over the years. We have clarified the relevant definition of complementarity of a pair of effects as nonexistence of a joint lower bound, by emphasising the operational interpretation of such a lower bound as a binary ``test'' observable. To define complementarity for observables, one then merely fixes a family of outcome sets for each observable; complementarity means that for any pair of sets from these respective families, the corresponding effects are complementary. We have presented several characterisations of complementarity, in terms of effect order, quantum operations implementing them, and their Naimark dilations, all appearing as consequences of an elementary lemma regarding factorisation of an effect in terms of a contraction, which is itself a reformulation of one of Paul's old results (as elaborated in \cite{BG1999}) on ``weak atoms'' of quantum effects. We have applied the characterisations to several cases, including position and momentum, position/momentum and energy, time and energy, number and phase, and spatial interferometry, which were all central to Paul's work. Regarding the noisy setting, we have discussed the complementarity as a form of extreme incompatibility of quantum observables, and finally considered the case of convolutions of position and momentum. In this context we have specifically focused on the complementarity of pairs of ``unsharp'' position and momentum effects where one of the functions is periodic and the other compactly supported, also settling an open question Paul posed in \cite{P2} as part of the work which originally initiated his study 
of circumventing complementarity and opening the scheme to develop a concept of unsharp reality.

\end{document}